\documentclass[sigconf]{acmart}
\setcopyright{none}
\renewcommand\footnotetextcopyrightpermission[1]{} 
\usepackage{graphicx}
\usepackage{tikz}
\usepackage{framed}
\usepackage{amsmath}
\usepackage{booktabs}
\usepackage{algorithm}
\usepackage{algorithmic}
\begin{document}

\title{Connection Incentives in Cost Sharing Mechanisms with Budgets}


\author{Tianyi Zhang, Dengji Zhao, Junyu Zhang, and Sizhe Gu}
\affiliation{%
  \institution{ShanghaiTech University}
  \city{Shanghai}
  \country{China}}
\email{{zhangty, zhaodj, zhangjy22022, guszh}@shanghaitech.edu.cn}





\begin{abstract}
In a cost sharing problem on a weighted undirected graph, all other nodes want to connect to the source node for some service. Each edge has a cost denoted by a weight and all the connected nodes should share the total cost for the connectivity. The goal of the existing solutions (e.g. folk solution and cycle-complete solution) is to design cost sharing rules with nice properties, e.g. budget balance and cost monotonicity. However, they did not consider the cases that each non-source node has a budget which is the maximum it can pay for its cost share and may cut its adjacent edges to reduce its cost share. In this paper, we design two cost sharing mechanisms taking into account the nodes' budgets and incentivizing all nodes to report all their adjacent edges so that we can minimize the total cost for the connectivity.
\end{abstract}



\maketitle

\section{Introduction}
Cost sharing problems are popular in economics and computer science~\cite{moulin2001strategyproof,leonardi2004cross,lorenzo2009characterization,DBLP:journals/mss/TrudeauV20}. There are a group of agents located at different points and a source. The agents need to connect to the source directly or indirectly in order to get a service or information~\cite{moulin1999incremental,gomez2017monotonic,todo2020split,DBLP:journals/igtr/BergantinosCLL22}. The cost of connecting any pair of agents is known and the total cost of connectivity has to be shared among all connected agents. This setting models many real-world applications from telephone and cable TV to water supply networks~\cite{gomez2011merge,trudeau2017set,DBLP:journals/ijgt/Trudeau23}.

To connect with the source, an agent may need to go through other agents. The intermediate agents can strategically block the connection if their cost share is reduced by doing so. However, such strategic behavior will potentially increase the total cost of connecting all agents. Therefore, it is important to design cost sharing mechanisms to prevent such strategic behavior so that we can minimize the total cost of the connectivity.

To solve the problem, we propose a cost sharing mechanism where each agent pays an average marginal cost to connect it to the source under the assumption that each agent's budget is unlimited. The reason for this assumption is that in some scenarios all the agents have enough budgets to pay their cost share. Since there exist some scenarios where the amount each agent can pay is limited, we also consider the situation where each agent has a limited budget for the share. Then we need to make sure that each agent can afford its cost share. One intuitive idea is that agents with larger budgets cover more costs than the agents with lower budgets. We need to give the agents with larger budgets incentives to do so, otherwise they will disconnect the agents with lower budgets. To that end, we determine their cost share in terms of their savings which are the difference between their budgets and the total cost.

\emph{Our contribution:} For cost sharing problems, to incentivize agents to share all their connections, we first design a mechanism where each agent's budget is unlimited that satisfies truthfulness (each agent is incentivized to offer all its connections) and other desirable properties studied in the literature. However, this mechanism does not satisfy budget feasibility (each agent's budget is larger than its cost share) when each agent's budget is limited. Therefore, we propose another mechanism that satisfies budget feasibility, truthfulness and other desirable properties.

\section{Related Work}
There is a rich literature for the cost sharing problem~\cite{2020Cooperative,DBLP:journals/anor/BergantinosL21,DBLP:journals/eor/HernandezPV23} (see the paper~\cite{2021A} for a comprehensive overview). 

Some treated the problem from the non-cooperative game perspective. Tijs and Driessen~\cite{tijs1986game} proposed the cost gap allocation (CGA) method based on marginal contribution. Bird~\cite{bird1976cost}, Dutta and Kar~\cite{dutta2004cost}, Norde \textit{et al.}~\cite{norde2004minimum}, Tijs \textit{et al.}~\cite{tijs2006obligation} and Hougaard \textit{et al.}~\cite{hougaard2010decentralized} provided cost sharing policies based on the minimum spanning tree~\cite{prim1957shortest} of a graph. However, they all cannot prevent the strategic behavior of agents since the agents can change the minimum spanning tree to reduce their cost share. 

Others treated the problem from the cooperative game perspective. They are all based on the Shapley value and differ in the definition of the coalition value. Kar~\cite{kar2002axiomatization} proposed the Kar solution and defined the value of a coalition $S$ as the minimal cost of connecting all agents of $S$ to the source without going through the agents outside of $S$. Berganti{\~n}os and Vidal-Puga~\cite{bergantinos2007fair} proposed the folk solution. They first computed the irreducible cost matrix, and then based on it, they defined the value of a coalition $S$ in the same way as the Kar solution. However, the folk solution throws away much information of the original graph. To look for a way to obtain a solution without throwing away as much information as the folk solution, Trudeau~\cite{trudeau2012new} proposed the cycle-complete solution. They made a less extreme transformation of the cost matrix.

These solutions all assumed that the agents in a coalition cannot use others outside of the coalition to connect to the source. Furthermore, Berganti{\~n}os and Vidal-Puga~\cite{bergantinos2007optimistic} introduced the optimistic game-based solution, where the value of a coalition $S$ is defined as the minimal cost of connecting all agents of $S$ to the source given that agents outside of $S$ are already connected to the source.

In summary, existing solutions are limited to complete graphs and do not consider the scenario where each agent has a budget to pay its cost share and has a strategic behavior. Therefore, they are inadequate for solving our problem.

\section{The Model}
\label{The model}
We consider a cost sharing problem to connect all the nodes on a weighted undirected graph $G=\langle V \cup \{s\},E \rangle$. Node $s$ is the source to which all the other nodes want to connect. The weight of each edge $(i,j) \in E$ denoted by $c_{(i,j)} \geq 0$ represents the cost to use the edge to connect $i$ with $j$. The total cost of the connectivity has to be shared among all connected nodes except for $s$. Each node $i \in V$ has a public budget $B_i > 0$ which is the maximum cost that $i$ can pay.

Consider a real scenario of the graph $G=\langle V \cup \{s\},E \rangle$, where $s$ represents a power station and each node $i \in V$ represents a village. A village can have electricity if there exists an electric transmission line connecting $s$ to the village. All villages require electricity and they need to share the total cost of the connectivity. Each village has a budget for the connectivity. 

Given the graph, the minimum cost of connecting all the nodes is the weight of a minimum spanning tree. The question here is how they share the minimum cost. We consider a natural strategic behavior that each node except for the source can cut its adjacent edges to share less cost. An edge $(i,j)$ cannot be used for the connectivity if either $i$ or $j$ cuts it. Our goal is to design cost sharing mechanisms to incentivize nodes to share all their adjacent edges so that we can use all the edges to minimize the total cost of the connectivity. 

Formally, let $\theta_i$ $(i \in V \cup \{s\})$ be the set of $i$'s adjacent edges and it is called $i$'s {\em type}. Let $\theta=(\theta_s,\theta_1,\cdots,\theta_{|V|})$ be the type profile of all nodes including the source $s$. We also write $\theta=(\theta_i,\theta_{-i})$, where $\theta_{-i}=(\theta_s,\theta_1,\cdots,\theta_{i-1},\theta_{i+1},\cdots,\theta_{|V|})$ is the type profile of all nodes except for $i$. Let $\Theta_i$ be the type space of $i$ and $\Theta=\Theta_s \times \Theta_1 \times \cdots \times \Theta_{|V|}$ be the type profile space of all nodes. 

We design cost sharing mechanisms that ask each node to report the set of its adjacent edges that can be used for the connectivity. Let $\theta_i' \subseteq \theta_i$ be the {\em report} of $i$, and $\theta'=(\theta_s,\theta_1',\cdots,\theta_{|V|}')$ be the report profile of all nodes. Note that the source does not behave strategically in this setting. Given a report profile $\theta' \in \Theta$, the graph induced by $\theta'$ is denoted by $G_{\theta'}=\langle V \cup \{s\} , E_{\theta'} \rangle \subseteq \langle V \cup \{s\}, E \rangle$, where $E_{\theta'} = \{(i,j)| (i,j) \in (\theta_i' \cap \theta_j')\} \subseteq E$.

\begin{definition}
\label{def2}
A cost sharing mechanism consists of a node selection policy $g: \Theta \rightarrow 2^V$, an edge selection policy $f: \Theta \rightarrow 2^E$, and a cost sharing policy $x = (x_i)_{i\in V}$ where $x_i: \Theta \rightarrow \mathbb{R}$. Given a report profile $\theta' \in \Theta$, $g(\theta')\subseteq V$ is the set of connected nodes, $f(\theta') \subseteq E_{\theta'}$ is the set of selected edges connecting the nodes in $g(\theta')$, and $x_i(\theta')$ is the cost share of $i$ (it will be zero if $i \notin g(\theta')$).
\end{definition}

For simplicity, we use $(g,f,x)$ to denote a cost sharing mechanism. In the following, we introduce the desirable properties of a cost sharing mechanism.

Truthfulness requires that for each node, its cost share for the connectivity is not decreased by cutting its adjacent edges. That is, offering all its adjacent edges minimizes its cost share.

\begin{definition}
\label{defic}
A cost sharing mechanism $(g,f,x)$ satisfies \textbf{truthfulness} if $$x_i((\theta_i, \theta'_{-i})) \leq x_i((\theta_i', \theta'_{-i})),$$ for all $i\in g(\theta_i, \theta'_{-i}) \cap g(\theta_i', \theta'_{-i})$, for all $\theta_i\in \Theta_i$, for all $\theta_i'\in \Theta_i$, and for all $\theta_{-i}'\in \Theta_{-i}=\Theta_s \times \Theta_1 \times \cdots \times \Theta_{i-1} \times \Theta_{i+1} \times \cdots \times \Theta_{|V|}$.
\end{definition}

Budget feasibility requires that each node's cost share is not over its budget. 
\begin{definition}
\label{defir}
A cost sharing mechanism $(g,f,x)$ satisfies \textbf{budget feasibility (BF)} if $$x_i(\theta') \leq B_i$$ for all $i \in V$, for all $\theta' \in \Theta$. 
\end{definition}

Cost monotonicity states that the cost share of a selected node will weakly increase if the cost of one of its adjacent edges increases within the same report profile. 

\begin{definition}
A cost sharing mechanism $(g,f,x)$ satisfies \textbf{cost monotonicity (CM)} if $$x_i(\theta') \leq x_i^+(\theta')$$ for all $\theta' \in \Theta$ and for all $i \in g(\theta')$, where $x_i^+(\theta')$ is the cost share of $i$ when the cost of one edge $(i,j) \in \theta_i'$ increases and $i$ is still selected.
\end{definition}

Budget balance states that the sum of all nodes' cost share equals the total cost of the selected edges for all report profiles. That is, the mechanism has no profit or loss.

\begin{definition}
A cost sharing mechanism $(g,f,x)$ satisfies \textbf{budget balance (BB)} if $$\sum_{i \in V}x_i(\theta')=\sum_{(i,j) \in f(\theta')}c_{(i,j)}$$ for all $\theta' \in \Theta$.
\end{definition}

Finally, positiveness states that each node's cost share should be non-negative.

\begin{definition}
A cost sharing mechanism $(g,f,x)$ satisfies \textbf{positiveness} if $$x_i(\theta')\geq 0$$ for all $i \in V$, for all $\theta' \in \Theta$.  
\end{definition}

In the rest of the paper, we design cost sharing mechanisms to satisfy the above properties.

\section{Cost Sharing with Unlimited Budget}
\label{csm-ub}
In this section, we first consider the case where each agent's budget is unlimited, i.e. they can always afford the cost share. We propose a mechanism that satisfies truthfulness, cost monotonicity, budget balance and positiveness.

The key ideas of the mechanism are as follows. First, we compute the minimum cost of connecting any subset of agents to the source. Then we get the marginal cost of adding each agent to a subset of other agents. The cost share of each agent is set to be the average marginal cost considering all possible joining sequences. The mechanism is formally described in Algorithm~\ref{alg}. A running example is given in Example~\ref{exp1}.

\begin{algorithm}[htb]
\caption{Average Marginal Cost Mechanism (AMCM)}
\label{alg}

\textbf{Input}:
A report profile $\theta' \in \Theta$\\
\textbf{Output}:
The selected nodes $g(\theta')$, the selected edges $f(\theta')$, 
the cost share $x(\theta')$
	\begin{algorithmic}[1]
	    \STATE Generate $G_{\theta'}$;
	    \STATE Starting from $s$, use Prim's algorithm~\cite{prim1957shortest} to compute a minimum spanning tree (MST) of $G_{\theta'}$;
	    \STATE Set $g(\theta')$ to be all the nodes of the MST and $f(\theta')$ to be all the edges of the MST;
		\FOR{each $S \subseteq g(\theta')$}
		\STATE Compute $v(S)$, which is the minimum cost to connect $S$ to $s$;
		\ENDFOR
		\FOR{each $i \in g(\theta')$}
		\STATE Compute $x_i(\theta')$,
		\begin{equation}
		\label{eq000}
		x_i(\theta')=\sum_{S \subseteq g(\theta') \backslash \{i\}}\frac{|S|!(|g(\theta')|-|S|-1)!}{|g(\theta')|!}
		    \cdot (v(S \cup \{i\})-v(S));
		\end{equation}
		\ENDFOR
		\FOR{each $i \in V \backslash g(\theta')$}
		\STATE $x_i(\theta')=0$;
		\ENDFOR
		\STATE \textbf{return} $g(\theta'),f(\theta'),x(\theta')$ 
	\end{algorithmic}
\end{algorithm}
\begin{figure}[htb]
    \centering
    \includegraphics[width=7.5cm]{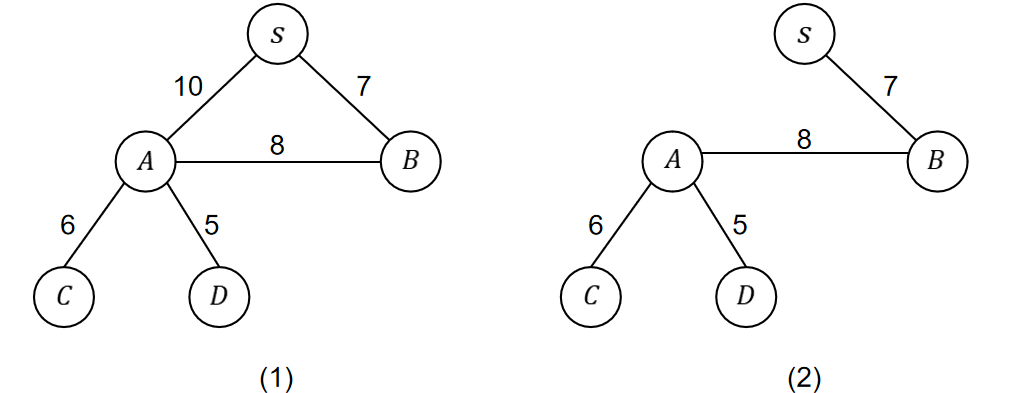}
    \caption{An example of the AMCM on a graph. The left figure is $G_{\theta'}$ and the right figure is the minimum spanning tree of $G_{\theta'}$. The letter $s$ represents the source, the letters $A,B,C,D$ represent the nodes, and the numbers on the edges represent the cost.}
	 \label{1111}
\end{figure}

\begin{example}
\label{exp1}
The graph $G_{\theta'}$ generated by a report profile $\theta' \in \Theta$ is shown in Figure~\ref{1111}(1). From the minimum spanning tree of $G_{\theta'}$ shown in Figure~\ref{1111}(2), we have $$g(\theta')=\{A,B,C,D\}$$ and $$f(\theta')=\{(s,B),(A,B),(A,C),(A,D)\}.$$ Then we compute the value function of each subset of $g(\theta')$. For example, we have $v(\emptyset)=0$, $v(\{A\})=10$, $v(\{A,B\})=15$, $v(\{A,B,C\})=21$ and $v(\{A,B,C,D\})=26$. Finally, by Equation~(\ref{eq000}), the cost share of nodes $A,B,C,D$ is $$x(\theta')=\left(\frac{19}{6},\frac{11}{2},\frac{55}{6},\frac{49}{6}\right).$$
\end{example}

\subsection{Properties of AMCM}
\label{pro}
\begin{theorem}
	\label{thethe3}
	The average marginal cost mechanism satisfies truthfulness. 
\end{theorem}

\begin{proof}
When $i$ truthfully reports its type, its cost share is
\begin{equation}
    \label{equa0}
		x_i(\theta')=\sum_{S \subseteq g(\theta') \backslash \{i\}}\frac{|S|!(|g(\theta')|-|S|-1)!}{|g(\theta')|!}
		    \cdot (v(S \cup \{i\})-v(S)),
    \end{equation}
		where $\theta'=(\theta_i,\theta_{-i}')$. 
      
      When $i$ misreports its type, its cost share is
     \begin{equation}
     \label{equa1}
		x_i(\theta'')=\sum_{S \subseteq g(\theta'') \backslash \{i\}}\frac{|S|!(|g(\theta'')|-|S|-1)!}{|g(\theta'')|!}
		    \cdot (v'(S \cup \{i\})-v'(S)),
     \end{equation}
	where $\theta''=(\theta_i',\theta_{-i}')$ and $v'(S)$ is the value function of $S$.
		    
    According to the definition of truthfulness, we need to show $x_i(\theta') \leq x_i(\theta'')$.   
    \begin{enumerate}
        \item When $g(\theta')=g(\theta'')$,
    for any given set $S \subseteq V$, we have
    \begin{equation}
    \label{equation00}
		\frac{|S|!(|g(\theta')|-|S|-1)!}{|g(\theta')|!} =\frac{|S|!(|g(\theta'')|-|S|-1)!}{|g(\theta'')|!}.
    \end{equation}
		When $i$ truthfully reports its type, there are two cases for computing $v(S)$. 
\begin{itemize}
    \item Node $i$ is not visited. If $i$ misreports its type (i.e. $\theta_i' \ne \theta_i$), then it is still not visited. So we have $v(S)=v'(S)$. According to the definition of $v$, we have $v'(S \cup \{i\}) \geq v(S \cup \{i\})$. Thus it holds that $v'(S \cup \{i\})-v'(S)\geq v(S \cup \{i\})-v(S)$.
    
    \item Node $i$ is visited. Then we have $v(S \cup \{i\})-v(S)=0$. If $i$ misreports its type, there are two possibilities.
	    \begin{itemize}
	        \item If it is still visited, then $v'(S \cup \{i\})-v'(S)=0$.
	        \item If it is not visited, by the monotonicity of value function, we have $v'(S\cup \{i\})-v'(S) \geq 0$.
	    \end{itemize}
\end{itemize}
	
	Therefore, we have $v'(S\cup \{i\})-v'(S) \geq v(S \cup \{i\})-v(S)$. Recalling to Equation~(\ref{equation00}), we infer that $x_i(\theta') \leq x_i(\theta'')$.
	
	\item When $g(\theta') \ne g(\theta'')$, then $g(\theta'')$ is a proper subset of $g(\theta')$. For each $S \subseteq g(\theta'')$, an argument similar to the one used in the case of $g(\theta')=g(\theta'')$ shows that the corresponding term in Equation~(\ref{equa0}) is less than or equal to that in Equation~(\ref{equa1}). Whereas for each $S \subseteq g(\theta') \backslash g(\theta'')$, the corresponding term in Equation~(\ref{equa0}) equals 0. So we have $x_i(\theta') \leq x_i(\theta'')$.
	\end{enumerate}
	\end{proof}

\begin{theorem}
\label{effi}
The average marginal cost mechanism satisfies budget balance. 
\end{theorem}
\begin{proof}
According to the edge selection policy, our mechanism outputs a spanning tree of the graph induced by $g(\theta') \cup \{s\}$. Given a report profile $\theta' \in \Theta$, by Equation~(\ref{eq000}) and the property of Shapley value, the sum of all nodes' cost share equals $v(g(\theta'))$, i.e. $$\sum_{i \in g(\theta')}x_i(\theta')=v(g(\theta')).$$ From the definition of value function, $v(g(\theta'))$ equals the total cost of minimum spanning tree of the graph induced by $g(\theta')\cup\{s\}$, i.e. $$v(g(\theta'))=\sum_{(i,j) \in f(\theta')}c_{(i,j)}.$$ Thus, $$\sum_{i \in g(\theta')}x_i(\theta')=\sum_{(i,j) \in f(\theta')}c_{(i,j)}.$$ So the proposed mechanism satisfies budget balance.
\end{proof}

\begin{theorem}
\label{}
The average marginal cost mechanism satisfies cost monotonicity.
\end{theorem}
\begin{proof}
Given $\theta' \in \Theta$ and nodes $i, j \in g(\theta')$, the cost share of $i$ is 
\begin{equation}
x_i(\theta')=\sum_{S \subseteq g(\theta') \backslash \{i\}}\frac{|S|!(|g(\theta')|-|S|-1)!}{|g(\theta')|!} \cdot (v(S \cup \{i\})-v(S)).   
\end{equation}

When $c_{(i,j)}$ does not change, there are two cases for computing $v(S)$.
\begin{itemize}
\item Node $i$ is not visited. If $c_{(i,j)}$ increases, $v(S)$ does not change, and $v(S \cup \{i\})$ increases or does not change. Thus $v(S \cup \{i\})-v(S)$ will increase or not change.

\item Node $i$ is visited. If $c_{(i,j)}$ increases, 
\begin{itemize}
    \item node $i$ is still visited, then $v(S \cup \{i\})-v(S)$ does not change.
    \item node $i$ is not visited, then $v(S \cup \{i\})-v(S)$ increases.  
\end{itemize}
\end{itemize}

In summary, when $c_{(i,j)}$ increases, $v(S \cup \{i\})-v(S)$ increases or does not change, and so does $x_i(\theta')$.
\end{proof}

\begin{theorem}
\label{}
The average marginal cost mechanism satisfies positiveness.
\end{theorem}

\begin{proof}
Given a report profile $\theta' \in \Theta$, for each $i \in V\backslash g(\theta')$, we have $x_i(\theta')=0$. For each $i \in g(\theta')$, by the definition of value function in Algorithm~\ref{alg}, we get $v(S \cup \{i\})-v(S) \geq 0$ for each $S \subseteq g(\theta')$. Therefore, according to Equation~(\ref{eq000}), each node's cost share $x_i(\theta')$ is non-negative.
\end{proof}

In the following example, we run the AMCM and the algorithm proposed by Claus and Kleitman~\cite{claus1973cost} on a tree. Note that the latter is only applied to trees, whereas our mechanism is applied to general graphs. 

\begin{figure}[htb]
    \centering
    \includegraphics[width=2.5cm]{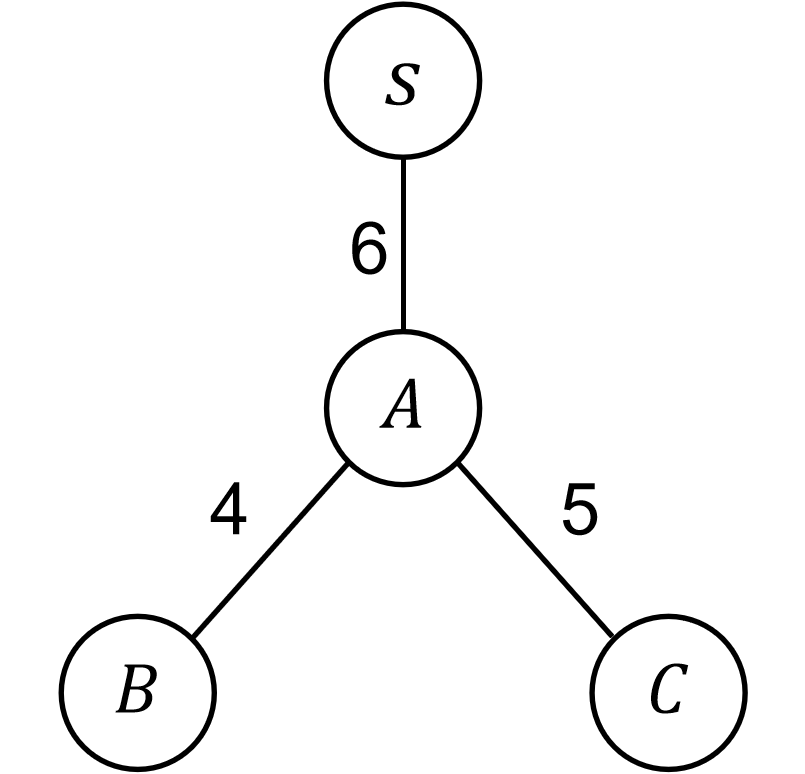}
     \caption{An example on a tree, where $s$ is the source and the numbers on the edges represent the cost.}
	 \label{Explanation1}
\end{figure}

\begin{example}
\label{exep21}
Figure~\ref{Explanation1} is generated by $\theta' \in \Theta$. We have $g(\theta')=\{A,B,C\}$. Then it follows that $v(\{A\})=6$, $v(\{B\})=10$, $v(\{C\})=11$, $v(\{A,B\})=10$, $v(\{A,C\})=11$, $v(\{B,C\})=15$, $v(\{A,B,C\})=15$ and $v(\emptyset)=0$. By Equation~(\ref{eq000}), we get $x_A(\theta')=2$, $x_B(\theta')=6$ and $x_C(\theta')=7$. Next, we compute each node's cost share under the algorithm proposed by Claus and Kleitman~\cite{claus1973cost}. According to the algorithm, the cost of each edge on the tree is equally shared among all nodes using it to connect to the source. Specifically, the cost of the edge $(s,A)$ is equally shared among $A,B,C$. The cost of the edge $(A,B)$ is shared by $B$. The cost of the edge $(A,C)$ is shared by $C$. Therefore, the cost share of $B$ is $x_B(\theta')=4+\frac{1}{3}\times 6=6$, the cost share of $C$ is $x_C(\theta')=5+\frac{1}{3}\times 6=7$ and the cost share of $A$ is $x_A(\theta')=\frac{1}{3}\times 6=2$.
\end{example}

In the above example, the cost share under AMCM equals that under the algorithm proposed by Claus and Kleitman~\cite{claus1973cost}. Actually, the two methods of computing cost share are equivalent on trees.

\begin{theorem}
\label{equva}
Each node's cost share under the AMCM coincides with that under the algorithm proposed by Claus and Kleitman~\cite{claus1973cost} on trees.
\end{theorem}

\begin{proof}
First, we prove that the statement holds on a line. Second, we extend the proof to a tree.

\begin{enumerate}
    \item For a given set $S \subseteq V$ on a line, if $S$ has the elements of $\{i,i+1,\cdots,|V|\}$, we have $v(S \cup \{i\})-v(S)=0$. Otherwise, $S \subseteq \{1,2,\cdots,i-1\}$.
 \begin{itemize}
     \item If $S=\emptyset$, then we have $$\phi_i=\frac{1}{|V|} (c_{s,1}+\cdots +c_{i-1,i}).$$
     \item If $S=j$ $(1 \leq j \leq i-1)$, then we have $$\phi_i=\sum_{j=1}^{i-1}\frac{1}{|V|(|V|-1)}(c_{(j,j+1)}+\cdots+c_{(i-1,i)}).$$
     \item If the largest element in $S$ is $2$, then $S=\{1,2\}$. Then $$\phi_i=\frac{2!(|V|-3)!}{|V|!}(c_{(2,3)}+\cdots+c_{(i-1,i)}).$$
     \item $\cdots$
     \item If the largest element in $S$ is $i-1$, we have $2 \leq |S| \leq i-1$. The $\phi_i$ is $$[\binom{1}{i-2}\frac{2!(|V|-3)!}{|V|!}+\cdots+\binom{i-2}{i-2}\frac{(i-1)!(|V|-i)!}{|V|!}]c_{(i-1,i)}.$$
 \end{itemize}

Putting the above analysis together, we obtain
 \begin{equation}
 \begin{split}
 \phi_i&=\frac{1}{|V|}\cdot (c_{(s,1)}+\cdots+c_{(i-1,i)})+\cdots+[\binom{1}{i-2}\frac{2!(|V|-3)!}{|V|!}+\\
 &\cdots+\binom{i-2}{i-2}\frac{(i-1)!(|V|-i)!}{|V|!}]c_{(i-1,i)}\\
 &=\frac{1}{|V|}c_{(s,1)}+(\frac{1}{n}+\frac{1}{|V|(|V|-1)})c_{(1,2)}+[\frac{1}{|V|}+\frac{1}{|V|(|V|-1)}+\\
 &\frac{1}{|V|(|V|-1)}+\frac{2!(|V|-3)!}{|V|!}]c_{(2,3)}+\cdots\\
 &=\frac{c_{(s,1)}}{|V|}+\frac{c_{(1,2)}}{|V|-1}+\frac{c_{(2,3)}}{|V|-2}+\cdots+\frac{c_{(i-1,i)}}{|V|-i+1}.    
 \end{split}
 \end{equation}

Therefore, we prove that the statement holds on a line. 
\item Consider a scenario where the nodes and their edges form a tree. We denote the path that node $i$ locates on is $l$ and there are $n$ nodes on this path. Denote $\overline{l}$ as the set of nodes that do not belong to this path. Then regarding the number of nodes in $\overline{l}$, we use mathematical induction to prove the statement.
 \begin{itemize}
     \item If $|\overline{l}|=1$, we assume $l=\{a\}$. Thus we have $$V=\{1,\cdots,i,\cdots,n,a\}.$$We let $N_i=\{1,\cdots,i-1,i+1,\cdots,n,a\}$. So all the subsets of $N_i$ can be divided into two parts: the subsets of $l_i=\{1,\cdots,i-1,i+1,\cdots,n\}$ and the set formed by adding $a$ to each subset of $l_i$. The power set of $l_i$ is denoted by $2^{l_i}=\{\emptyset, \{1\}, \cdots, \{i-1\}, \{i+1\},\cdots, \{n\}, \cdots, \{1,\cdots,i-1,i+1,\cdots,n\}\}$. The set formed by adding $a$ to each subset of $l_i$ is denoted by $(2^{l_i})^C=\{\{a\}, \{1,a\}, \cdots, \{i-1,a\}, \{i+1,a\},\cdots, \{n,a\}, \cdots, \{1,\cdots,i-1,i+1,\cdots,n,a\}\}$. Next, we analyze each element of $2^{l_i}$ and compute $\phi_i$.
     \begin{itemize}
         \item For $\emptyset$: we have $$\phi_i=\frac{1}{|V|+1}(c_{(s,1)}+\cdots+c_{(i-1,i)}).$$
         \item $\cdots$
         \item For $\{i-1\}$: we have $$\phi_i=\frac{1}{(|V|+1)|V|}c_{(i-1,i)}.$$ 
     \end{itemize}
     For other subsets, we use the following method to compute $\phi_i$.
     \begin{itemize}
         \item The subset with $2$ as the largest element, i.e. $\{1,2\}$. We have $$\phi_i=\frac{2!(|V|-2)!}{(|V|+1)!}(c_{(2,3)}+\cdots+c_{(i-1,i)}).$$
         \item $\cdots$
         \item The subset with $i-1$ as the largest element. We have $$\phi_i=[\binom{1}{i-2}\frac{2!(|V|-2)!}{(|V|+1)!}+\cdots+\binom{i-2}{i-2}\frac{(i-1)!(|V|-i+1)!}{(|V|+1)!}]c_{(i-1,i)}.$$ 
     \end{itemize}
    
     Therefore, the sum of $\phi_i$ is:
     \begin{equation}
         \begin{split}
             &\frac{1}{|V|+1}(c_{(s,1)}+\cdots+c_{(i-1,i)})+\cdots+
             [\binom{1}{i-2}\frac{2!(|V|-2)!}{(|V|+1)!}+\\
             &\cdots+\binom{i-2}{i-2}\frac{(i-1)!(|V|-i)!}{(|V|+1)!}]c_{(i-1,i)}\\
             &=\frac{1}{|V|+1}c_{(s,1)}+\frac{1}{|V|}c_{(1,2)}+\frac{1}{|V|-1}c_{(2,3)}+\cdots+\frac{1}{|V|+2-i}c_{(i-1,i)}.
         \end{split}
     \end{equation}
     For all the subsets of $(2^{l_i})^C$, similarly, we get that the sum of $\phi_i$ is:
     \begin{equation}
         \begin{split}
          &\frac{1}{|V|(|V|+1)}c_{(s,1)}+[\frac{1}{|V|(|V|+1)}+\frac{2}{|V|(|V|-1)(|V|+1)}]c_{(1,2)}+\cdots\\
          &=\frac{1}{|V|(|V|+1)}c_{(s,1)}+\frac{1}{|V|(|V|-1)}c_{(1,2)}+\cdots.
         \end{split}
     \end{equation}
    
     At last, we have:
     \begin{equation}
         \begin{split}
             \phi_i&=[\frac{1}{|V|+1}+\frac{1}{|V|(|V|+1)}]c_{(s,1)}+[\frac{1}{|V|}+\frac{1}{|V|(|V|-1)}]\\
             &c_{(1,2)}+[\frac{1}{|V|-1}+\frac{1}{(|V|-1)(|V|-2)}]           c_{(2,3)}+\cdots\\
             &=\frac{c_{(s,1)}}{|V|}+\frac{c_{(1,2)}}{|V|-1}+\cdots+\frac{c_{(i-1,i)}}{|V|-i+1}.
         \end{split}
     \end{equation}
     That is, when $|\overline{l}|=1$, the statement holds.  
     \item Assume that the statement holds for $|\overline{l}|=k$. We need to prove that it holds for $|\overline{l}|=k+1$. We assume that the newly added node is $b$ and the set of nodes not belonging to the set $\{1,\cdots,n,b\}$ is $\overline{l_k}$. For each node $i$, we have
     \begin{equation}
     \begin{split}
     \phi_i&=\sum_{S \subseteq (l \cup \overline{l_k} \backslash  \{i\})}\frac{|S|!(|V|+k-|S|)!}{(k+1+|V|)!}(v(S \cup \{i\})-v(S))+\\
     &\sum_{S \cup \{b\}}\frac{(|S|+1)!(|V|+k-|S|-1)!}{(k+1+|V|)!}(v(S \cup \{i\})-v(S)). 
     \end{split}
     \end{equation}
     Let $$A=\frac{|S|!(|V|+k-|S|)!}{(k+1+|V|)!}$$ and $$B=\frac{(|S|+1)!(|V|+k-|S|-1)!}{(k+1+|V|)!}.$$Then we get $$A+B=\frac{|S|!(|V|+k-1-|S|)!}{(|V|+k)!}.$$Therefore, $$\phi_i=\sum_{S \subseteq (l \cup \overline{l_k} \backslash \{i\})}\frac{|S|!(|V|+k-|S|-1)!}{(k+|V|)!}(v(S \cup \{i\})-v(S)).$$ According to inductive hypothesis, we get 
     \begin{equation}
     \begin{split}
    &\sum_{S \subseteq (l \cup \overline{l_k} \backslash \{i\})}\frac{|S|!(|V|+k-|S|-1)!}{(k+|V|)!}(v(S \cup \{i\})-v(S))\\
    &=\frac{c_{(s,1)}}{|V|}+\frac{c_{(1,2)}}{|V|-1}+\frac{c_{(2,3)}}{|V|-2}+\cdots.
     \end{split}
     \end{equation}
     Thus we have $$\phi_i=\frac{c_{(s,1)}}{|V|}+\frac{c_{(1,2)}}{|V|-1}+\frac{c_{(2,3)}}{|V|-2}+\cdots.$$ That is, the statement holds for $|\overline{l}|=k+1$.
 \end{itemize}
\end{enumerate}
\end{proof}

\section{Cost Sharing with Limited Budget}
\label{csm-lb}
Now we consider the situation that each node has a limited budget. We first show that the AMCM defined in Algorithm~\ref{alg} does not satisfy budget feasibility by the example illustrated in Figure~\ref{tu1}. For nodes $C$ and $D$, since $x_C(\theta')=\frac{55}{6}>7= B_C$ and $x_D(\theta')=\frac{49}{6}>6=B_D$, budget feasibility is violated. So in this section, we propose a new cost sharing mechanism to satisfy budget feasibility.

For the mechanism to satisfy budget feasibility (i.e. the budget of each node should cover its cost share), we design a special node selection function defined in Algorithm~\ref{select}. When given a graph generated by a report profile and a set of nodes, the function outputs a subset of nodes connected to the source. To satisfy budget feasibility, we have to remove those nodes from the final connection whose budgets are lower than their cost share. In other words, a node is added to the group of nodes that are connected to the source if its budget covers the minimal cost of the edges connecting to the group. Following this principle, starting from $s$, we gradually add nodes into the connected group until no more nodes can be added.

\begin{algorithm}[htb]
	\caption{Selection Function $g^*$($\theta'$,$S$)}
	\label{select}

\textbf{Input}: 
A report profile $\theta' \in \Theta$,
a set $S \subseteq V$\\   

\textbf{Output}:
A subset of $S$ connected to the source 
	\begin{algorithmic}[1]
        \STATE Generate $G_{\theta'}$;
	    \STATE Initialize $N=$ $\{s\}$;
	    \STATE Set $\hat{E} = \{(p,q)| p \in N, q \in S\} \subseteq E_{\theta'}$;
	    \WHILE{$\hat{E}$ is not empty}
	    \STATE Find the edge $(i,j)$ in $\hat{E}$ with the smallest weight where $i \in N$ and $j \in S \backslash N$;
	    \IF{$B_j \geq c_{(i,j)}$}
	    \STATE $N$ $=$ $N \cup \{j\}$;
	    \STATE $\hat{E} = \{(p,q)| p \in N, q \in S \backslash N\}$;
	    \ELSE 
	    \STATE $\hat{E} =$ $\hat{E} \backslash \{(i,j)\}$;
	    \ENDIF
	    \ENDWHILE
	    \STATE \textbf{return} $N \backslash \{s\}$
	\end{algorithmic}
\end{algorithm}

\begin{algorithm}[tb]
	\caption{Saving-based Cost Sharing Mechanism (SCSM)}
	\label{budget}

\textbf{Input}:  
A report profile $\theta' \in \Theta$ 

\textbf{Output}:
The node selection $g(\theta')$, the edge selection $f(\theta')$, 
the cost share $x(\theta')$
	\begin{algorithmic}[1]
	    \STATE Set $g(\theta')$ = $g^*(\theta',V$);
	    \FOR{each $S \subseteq g(\theta')$}
	    \STATE Let $S' = g^*(\theta',S$);
	    \STATE Set 
     \begin{equation}
     \label{sss}
       v(S)=\sum_{i \in S'}B_i-C(S'),  
     \end{equation}
     where $C(S')$ is the minimal cost of connecting $S'$ to $\{s\}$;
	    \ENDFOR
	    \FOR{each $i \in g(\theta')$}
        \STATE Compute $\phi_i$, 
        \begin{equation}
        \label{scs1}
		\phi_i=\sum_{S \subseteq g(\theta') \backslash \{i\}}\frac{|S|!(|g(\theta')|-|S|-1)!}{|g(\theta')|!}
		    \cdot (v(S \cup \{i\})-v(S));  
        \end{equation}
	    \STATE Let 
        \begin{equation}
        \label{scs2}
                x_i(\theta')=B_i-\phi_i; 
        \end{equation}     
	    \ENDFOR
	    \FOR{each $i \in V\backslash g(\theta')$}
	    \STATE $x_i(\theta')=0$;
	    \ENDFOR
	    \STATE Set $f(\theta')=E^{MST}_{g(\theta')}$, where $E^{MST}_{g(\theta')}$ is the set of edges of minimum spanning tree of the graph induced by $g(\theta') \cup \{s\}$;
		\STATE \textbf{return} $g(\theta'),f(\theta'),x(\theta')$ 
	\end{algorithmic}
\end{algorithm}

\begin{figure}
    \centering
    \includegraphics[width=7.5cm]{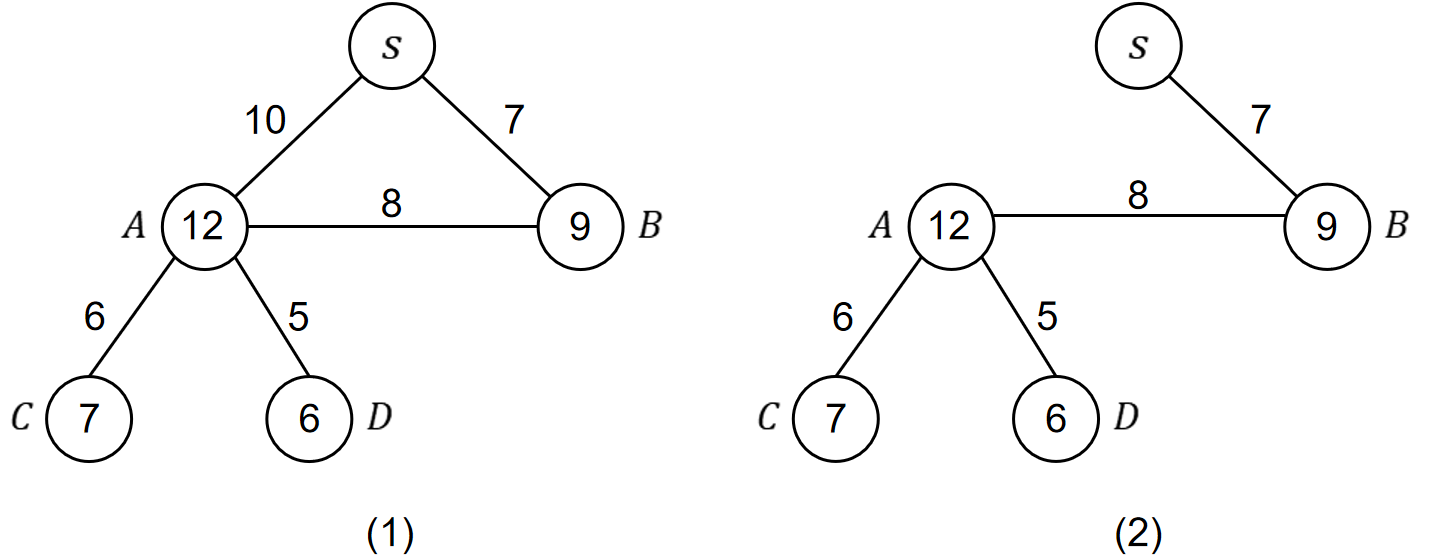}
     \caption{An example of the SCSM on a graph. The left figure is $G_{\theta'}$ and the right figure is the minimum spanning tree of $G_{\theta'}$. The letter $s$ is the source, the numbers on the edges represent the cost and the number in the circle represents the budget of the node. 
     }
	 \label{tu1}
\end{figure}

After choosing the connected nodes, can we simply apply the AMCM to compute their cost share to meet their budgets? The answer is no. The example given in Figure~\ref{tu1} can be used to explain the reason. By Algorithm~\ref{select}, nodes $C$ and $D$ are selected, but according to the definition of AMCM, they will pay more than their budgets. 

To satisfy budget feasibility, we incorporate the nodes' budgets into the computation of their cost share. For the set of selected nodes, we compute its savings, which are the difference between their budgets and the total cost. Then we allocate the savings among the nodes in the set so that those nodes with lower budgets have lower cost shares. 

The new mechanism is called saving-based cost sharing mechanism (SCSM) and is formally described in Algorithm~\ref{budget}. And, it is illustrated by Example~\ref{exep2}.

\begin{example}
\label{exep2}
The graph $G_{\theta'}$ generated by a report profile $\theta' \in \Theta$ is shown in Figure~\ref{tu1}(1). According to Algorithm~\ref{select}, first we compare $B_B$ with $c_{(s,B)}$. Since $B_B > c_{(s,B)}$, node $B$ is selected by the mechanism. Second, since $B_A > c_{(A,B)}$, node $A$ is selected by the mechanism. Similarly, nodes $C$ and $D$ are also selected by the mechanism. So the set of nodes selected by the mechanism is $g(\theta')=\{A,B,C,D\}$ (see Figure~\ref{tu1}(2)). For each $S \subseteq g(\theta')$, we compute $g^*(\theta',S)$ by Algorithm~\ref{select}. For example, when $S=\{B,C\}$, we have $g^*(\theta',S)=\{B\}$. Further, by Equation~(\ref{sss}), we get the value function of $S$. For example, we have $v(\{B,C\})=\sum_{i \in g^*(\theta',\{B,C\})}B_i-C(g^*(\theta',\{B,C\}))=2$. By Equation~(\ref{scs1}) and Equation~(\ref{scs2}), we compute the cost shares of $A,B,C,D$, i.e. $x(\theta')=(8,6,6.5,5.5)$. From Figure~\ref{tu1}(2), the set of selected edges is $$f(\theta')=\{(s,B),(A,B),(A,C),(A,D)\}.$$
\end{example}

In the following example, we run the SCSM and another method using savings on a tree.
\begin{figure}[htb]
    \centering
    \includegraphics[width=3cm]{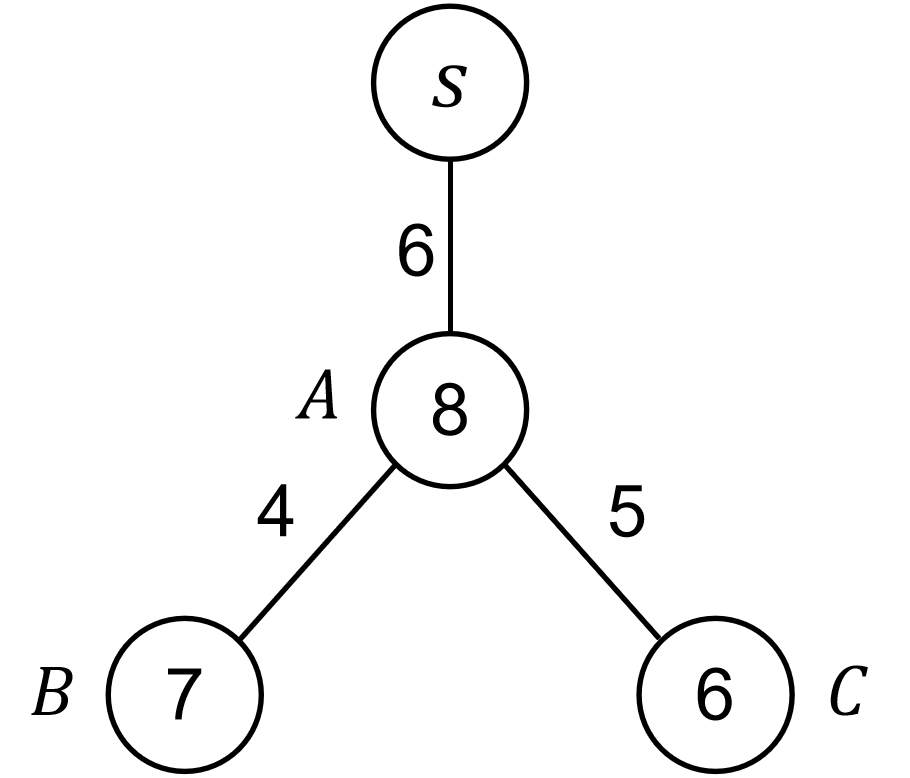}
     \caption{An example on a tree, where $s$ is the source, the numbers on the edges represent the cost and the number in each circle represents the budget of the node.}
	 \label{Explanation2}
\end{figure}
\begin{example}
\label{exep22}
Figure~\ref{Explanation2} is generated by a report profile $\theta' \in \Theta$. Using Algorithm~\ref{select}, we get $g(\theta')=\{A,B,C\}$. It follows that $v(\emptyset)=0$, $v(\{A\})=2$, $v(\{B\})=0$, $v(\{C\})=0$, $v(\{A,B\})=5$, $v(\{A,C\})=3$, $v(\{B,C\})=0$ and $v(\{A,B,C\})=6$. So we have $\phi_A=4, \phi_B=1.5 ,\phi_C=0.5$, and thus $x_A(\theta')=4$, $x_B(\theta')=5.5$, $x_C(\theta')=5.5$. Then we compute each node's cost share using its savings. For node $B$, its savings $B_B-c_{(A,B)}$ is shared among $A$ and $B$ equally. Thus, its cost share is $x_B(\theta') = B_B-\frac{1}{2}\cdot (B_B-c_{(A,B)})=7-\frac{7-4}{2} =5.5$. Similarly, we have $x_C(\theta') = B_C-\frac{1}{2}\cdot (B_C-c_{(A,C)})=6-\frac{6-5}{2} = 5.5$. For node $A$, its savings are not shared among other nodes. So its cost share is $x_A(\theta') = B_A-(B_A-c_{s,A})-\frac{1}{2}\cdot (B_C-c_{(A,C)})-\frac{1}{2}\cdot (B_B-c_{(A,B)})=8-(8-6)-\frac{6-5}{2}-\frac{7-4}{2}=4$. 
\end{example}

In the above example, the cost share under SCSM equals that under the method using savings. In fact, the two methods of computing cost share are equivalent on trees. 

Next, we compare the AMCM with the SCSM. Under AMCM, each node should share the total cost of the path from itself to the source. Since the AMCM does not consider its budget at all, its cost share might be larger than its budget. Under SCSM, each node's savings (the difference between its budget and the cost of the edge connecting it to its parent) should be shared among all its ancestors. In the worst case, the savings allocated to a node is zero. Therefore, each node's cost share is not more than its budget.

\subsection{Properties of SCSM}
\begin{theorem}
\label{truthful11}
The saving-based cost sharing mechanism satisfies truthfulness.
\end{theorem}
\begin{proof}
    When node $i$ truthfully reports its type $\theta_i$ (i.e. $\theta' = (\theta_i, \theta_{-i}')$), its cost share is  
    $x_i(\theta')=B_i-\phi_i$, where\\
    \begin{equation}
    \label{equa2}
		\phi_i=\sum_{S \subseteq g(\theta') \backslash \{i\}}\frac{|S|!(|g(\theta')|-|S|-1)!}{|g(\theta')|!}
		    \cdot (v(S \cup \{i\})-v(S)).
	\end{equation}
    
    When node $i$ misreports its type $\theta_i$ (i.e. $\theta'' = (\theta_i', \theta_{-i}')$), its cost share is
    $x_i(\theta'')=B_i-\phi_i'$, where\\
    \begin{equation}
    \label{equa3}
		\phi_i'=\sum_{S \subseteq g(\theta'') \backslash \{i\}}\frac{|S|!(|g(\theta'')|-|S|-1)!}{|g(\theta'')|!}
		    \cdot (v'(S \cup \{i\})-v'(S)),
		\end{equation}
		 with $v'(S)$ being the value function of $S$. All we need to do is to show $\phi_i \geq \phi_i'$.

\begin{enumerate}
    \item When $g(\theta')=g(\theta'')$,
    for each $S \subseteq g(\theta') \backslash \{i\}$, we have
    \begin{equation*}
		\frac{|S|!(|g(\theta')|-|S|-1)!}{|g(\theta')|!} =\frac{|S|!(|g(\theta'')|-|S|-1)!}{|g(\theta'')|!}.
    \end{equation*}
    
		Next, we compare $(v(S \cup \{i\})-v(S))$ with $(v'(S \cup \{i\})-v'(S))$ in three cases.
    \begin{itemize}		
\item $g^*(\theta',S\cup\{i\})=g^*(\theta',S)\cup\{i\}$. In this case, $i$ is selected and there are no more nodes in $S$ being selected due to $i$'s participation. When $i$ truthfully reports its type, we have
\begin{equation*}
v(S\cup\{i\})-v(S) =B_i+C(g^*(\theta',S))-C(g^*(\theta',S)\cup\{i\}).
\end{equation*}
Since there is at least one edge connecting $i$ to $g^*(\theta',S)\cup\{s\}$ and $i$ can afford the cost of the edge, we have
\begin{equation*}
B_i+C(g^*(\theta',S))-C(g^*(\theta',S)\cup\{i\}) \geq 0.
\end{equation*}
When $i$ misreports its type, there are two possibilities.
	    	 \begin{itemize}        
	    	    \item $g^*(\theta'',S\cup\{i\})=g^*(\theta',S)\cup\{i\}$: since $C(g^*(\theta',S))$ cannot use the nodes outside $g^*(\theta',S)$, it remains unchanged. Again, due to $i$'s misreporting, we have $$C(g^*(\theta'',S)\cup\{i\}) \geq C(g^*(\theta',S)\cup\{i\}).$$Hence, we infer that $v(S\cup\{i\})-v(S)\geq v'(S\cup\{i\})-v'(S)$. 
	            \item $g^*(\theta'',S\cup\{i\})=g^*(\theta',S)$: we have $v(S\cup\{i\})-v(S)\geq v'(S\cup\{i\})-v'(S)=0$.  
	        \end{itemize}
	        
\item $g^*(\theta',S\cup\{i\})=g^*(\theta',S)$. In this case, when $i$ truthfully reports its type, it is not selected and we get $v(S\cup\{i\})-v(S)=0$. When node $i$ misreports its type, it is still not selected since the set of available edges shrinks. Hence, we have $v(S\cup\{i\})-v(S)=v'(S\cup\{i\})-v'(S)=0$.  

\item $g^*(\theta',S\cup\{i\})=g^*(\theta',S)\cup  S_0\cup\{i\}$, where $S_0$ denotes the set of nodes that are selected in $g^*(\theta',S\cup\{i\})$ due to $i$'s participation. When $i$ truthfully reports its type, we have
\begin{equation*}
\begin{split}
&v(S\cup\{i\})-v(S)\\
&=\sum_{j\in S_0\cup\{i\}}B_j +C(g^*(\theta',S)) - C(g^*(\theta',S)\cup  S_0\cup\{i\}). \\   
\end{split}
\end{equation*}
The cost of the edges which belong to the spanning tree generated by Algorithm~\ref{select} and which connect $S_0 \cup \{i\}$ to $s$ is larger than or equal to $C(g^*(\theta',S) \cup S_0 \cup \{i\})-C(g^*(\theta',S))$. Hence, we get $v(S\cup\{i\})-v(S)\geq 0 $. When node $i$ misreports its type, there are two possibilities.
	    \begin{itemize}
	        \item $g^*(\theta'',S\cup\{i\})=g^*(\theta',S)$: since node $i$ is not selected by our mechanism, it holds that $v(S\cup \{i\})-v(S)\geq v'(S\cup\{i\})-v'(S)=0$. 
	        \item $g^*(\theta'',S\cup\{i\})=g^*(\theta',S)\cup  S_0^{'}\cup\{i\}$ with $S_0^{'} \subseteq S_0$: node $i$ is still selected and its misreporting may shrink the set $S_0$. Let 
	        \begin{equation*}
	        \begin{split}
	       m_i &= v(S\cup \{i\})-v(S)\\&= 
	       \sum_{j\in S_0\cup\{i\}}B_j +C(g^*(\theta',S)) - C(g^*(\theta',S)\cup  S_0\cup\{i\})
	        \end{split}
	        \end{equation*}
	         and 
	         \begin{equation*}
	             \begin{split}
	               m_i' &= v'(S\cup \{i\})-v'(S)\\ 
	               &= \sum_{j\in S_0^{'}\cup\{i\}}B_j +C(g^*(\theta'',S)) - C(g^*(\theta'',S)\cup  S_0^{'}\cup\{i\}).
	             \end{split}
	         \end{equation*}
	          We infer that	 
	          \begin{equation*}
	              \begin{split}
	       &m_i-m_i'  \\&= 
	       \sum_{j\in S_0\backslash S_0^{'}}B_j - C(g^*(\theta',S)\cup  S_0\cup\{i\}) + C(g^*(\theta'',S)\cup  S_0^{'}\cup\{i\})\\
	       &= \sum_{j\in S_0\backslash S_0^{'}}B_j-C(g^*(\theta',S)\cup S_0\cup\{i\}) +C(g^*(\theta',S)\cup  S_0^{'}\cup\{i\})
	       \\
	       &-C(g^*(\theta',S)\cup  S_0^{'}\cup\{i\})+C(g^*(\theta'',S)\cup  S_0^{'}\cup\{i\}).
	              \end{split}
	          \end{equation*}
	       Since the nodes in $S_0\backslash S_0^{'}$ are selected when node $i$ reports its type truthfully, we get
	       \begin{equation*}
	       \sum_{j\in S_0\backslash S_0^{'}}B_j +C(g^*(\theta',S)\cup  S_0^{'}\cup\{i\})-C(g^*(\theta',S)\cup S_0\cup\{i\}) \geq 0.
	       \end{equation*}
	        For the nodes in $g^*(\theta',S)\cup S_0^{'}\cup\{i\}$, the cost of the minimum spanning tree of the graph induced by $g^*(\theta',S)\cup S_0^{'}\cup\{i\}$ will increase when node $i$ misreports its type since the set of available edges shrinks. Then we have
	        \begin{equation*}
	         C(g^*(\theta'',S)\cup  S_0^{'}\cup\{i\})-C(g^*(\theta',S)\cup  S_0^{'}\cup\{i\}) \geq 0.   
	        \end{equation*}
	         Thus $m_i-m_i' \geq 0$. Further, we have $v(S\cup \{i\})-v(S)\geq v'(S\cup\{i\})-v'(S)$. 
	    \end{itemize}
\end{itemize}
	Putting the above analysis together, we have $\phi_i \geq \phi_i'$.
	\item When $g(\theta') \ne g(\theta'')$, then $g(\theta'')$ is a proper subset of $g(\theta')$. For each $S \subseteq g(\theta'')$, an argument similar to the one used in the case of $g(\theta')=g(\theta'')$ shows that the corresponding term in Equation~(\ref{equa2}) is larger than or equal to that in Equation~(\ref{equa3}). Whereas for each $S \subseteq g(\theta') \backslash g(\theta'')$, the corresponding term in Equation~(\ref{equa2}) is larger than or equal to 0. So we have $\phi_i \geq \phi_i'$.
\end{enumerate}
\end{proof}

\begin{theorem}
\label{}
The saving-based cost sharing mechanism satisfies budget feasibility.
\end{theorem}
\begin{proof}
    Given a report profile $\theta' \in \Theta$, for $i \in V \backslash g(\theta')$, its cost share $x_i(\theta')=0$, which is less than its budget. For $i \in g(\theta')$, we have $x_i(\theta') = B_i -\phi_i$, where
    \begin{equation*}
		\phi_i=\sum_{S \subseteq g(\theta') \backslash \{i\}}\frac{|S|!(|g(\theta')|-|S|-1)!}{|g(\theta')|!}
		    \cdot (v(S \cup \{i\})-v(S)).
	\end{equation*}
    By the definition of selection function, there are three cases.
\begin{itemize}
\item $g^*(\theta',S\cup\{i\})=g^*(\theta',S)$. Since node $i$ is not selected, we have $v(S\cup\{i\})-v(S)=0$.

\item $g^*(\theta',S\cup\{i\})=g^*(\theta',S)\cup\{i\}$. Since node $i$ is selected, we have $v(S\cup\{i\})-v(S)=B_i+C(g^*(\theta',S))-C(g^*(\theta',S)\cup\{i\})$. From the proof of Theorem~\ref{truthful11}, we have $B_i+C(g^*(\theta',S))-C(g^*(\theta',S)\cup\{i\}) \geq 0$ and thus $v(S\cup\{i\})-v(S)\geq 0$.

\item $g^*(\theta',S\cup\{i\})=g^*(\theta',S)\cup  S_0\cup\{i\}$, where $S_0$ denotes the set of nodes that are selected in $g^*(\theta',S\cup\{i\})$ due to $i$'s participation. From the proof of Theorem~\ref{truthful11}, we get $v(S\cup\{i\})-v(S)\geq 0$.
\end{itemize}
    From the above analysis, it follows that for node $i \in g(\theta')$, $\phi_i$ is non-negative. So we infer that $x_i(\theta') \leq B_i$.
\end{proof}

\begin{theorem}
\label{}
The saving-based cost sharing mechanism satisfies budget balance.
\end{theorem}

\begin{proof}
    According to the edge selection policy, our mechanism outputs a spanning tree of the graph induced by $g(\theta')\cup \{s\}$. Given a report profile $\theta' \in \Theta$, by $x_i(\theta')=B_i-\phi_i$, we have
    \begin{equation*}
    \sum_{i\in g(\theta')}x_i(\theta')=\sum_{i\in g(\theta')}B_i-\sum_{i\in g(\theta')}\phi_i.
    \end{equation*}
    By the property of Shapley value, we have 
    \begin{equation*}
    \sum_{i\in g(\theta')}\phi_i = v(g(\theta')). 
    \end{equation*}
    Therefore, 
    \begin{equation*}
     \sum_{i\in g(\theta')}x_i(\theta') =  \sum_{i\in g(\theta')}B_i - v(g(\theta')).   
    \end{equation*}
    By the definition of mechanism, we get
    \begin{equation*}
       v(g(\theta')) =  \sum_{i\in g(\theta')}B_i - C(g(\theta')). 
    \end{equation*}
    Thus, we get
    \begin{equation*}
    \begin{split}
    \sum_{i\in g(\theta')}x_i(\theta')&=  \sum_{i\in g(\theta')}B_i-(\sum_{i\in g(\theta')}B_i-C(g(\theta')))\\ 
    &=C(g(\theta'))\\
    &=\sum_{(i,j) \in f(\theta')}c_{(i,j)},
    \end{split}
    \end{equation*}
     which states that our mechanism satisfies budget balance. 
\end{proof}

\begin{theorem}
\label{}
The saving-based cost sharing mechanism satisfies cost monotonicity.
\end{theorem}

\begin{proof}
Given $\theta' \in \Theta$, for all $i \in g(\theta')$, for all $(i,j) \in E$, when $c_{(i,j)}$ increases, there are two possibilities.
		    
First, $g(\theta')$ does not change. Since $x_i(\theta')=B_i-\phi_i$ where
\begin{equation}
		\phi_i=\sum_{S \subseteq g(\theta') \backslash \{i\}}\frac{|S|!(|g(\theta')|-|S|-1)!}{|g(\theta')|!} \cdot (v(S \cup \{i\})-v(S)),   
\end{equation}
it suffices to show that $v(S\cup\{i\})-v(S)$ decreases or does not change.
 \begin{itemize}
\item If $i \notin g(\theta')$, then $i$ is still not selected by the mechanism. Thus $v(S \cup \{i\})-v(S)$ does not change.

\item If $i \in g(\theta')$, for node $i$, there are two cases.
    \begin{itemize}
        \item Node $i$ is not selected anymore. Its cost share becomes 0. So $v(S \cup \{i\})-v(S)$ decreases.
        \item Node $i$ is still selected. For each $S \subseteq V$, by the expression of $v(S\cup\{i\})-v(S)$, it decreases.
    \end{itemize}
    \end{itemize}
    
Second, $g(\theta')$ shrinks. We show that the cost share of $i$ weakly increases. 
    \begin{itemize}
        \item If $i \notin g(\theta')$, its cost share does not change and the statement holds. 
        \item If $i \in g(\theta')$, then $j \notin g(\theta')$. We infer that $j$ connects to the source through $(i,j)$ before $g(\theta')$ shrinks. Therefore, $i$'s savings decrease. Thus the cost share of $i$ weakly increases.    
 \end{itemize}
 
Putting the above analysis together, when $c_{(i,j)}$ increases, the cost share of $i$ weakly increases.
\end{proof}

\section{Conclusions}
\label{con}
For cost sharing problems on graphs, we considered an important strategic behavior of agents, i.e. cutting their adjacent edges. We proposed two mechanisms where each agent has an unlimited budget or a limited budget respectively. Our mechanisms can incentivize all agents to offer all their adjacent edges so that we can minimize the total cost of connectivity by using all the connections. One meaningful future work is to characterize all possible such mechanisms.

\begin{acks}
This work is supported by Science and Technology Commission of Shanghai Municipality (No. 23010503000 and No. 22ZR1442200), and Shanghai Frontiers Science Center of Human-centered Artificial Intelligence (ShangHAI).
\end{acks}

\bibliographystyle{ACM-Reference-Format}
\bibliography{sample-base}

\end{document}